\documentclass[a4paper, 10pt]{article}
\usepackage[utf8]{inputenc}

\title{Generalized column distances of convolutional codes}
\author{Elisa Gorla and Flavio Salizzoni \\
{\small {\it Dedicated to our friend and mentor Joachim Rosenthal on the occasion of his sixtieth birthday.}}}

\date{}

\usepackage{geometry}
\usepackage{amssymb,amsmath,amsthm}
\usepackage{graphicx,url,comment,xcolor}

\theoremstyle{definition}
\newtheorem{theorem}{Theorem}[section]
\newtheorem{proposition}[theorem]{Proposition}
\newtheorem{lemma}[theorem]{Lemma}
\newtheorem{definition}[theorem]{Definition}
\newtheorem{example}[theorem]{Example}
\newtheorem{remark}[theorem]{Remark}
\newtheorem{corollary}[theorem]{Corollary}

\newcommand{\N}{\mathbb N}

\newcommand{\Z}{\mathbb Z}

\newcommand{\F}{\mathbb F}

\newcommand{\Cc}{\mathcal C}
\newcommand{\Dd}{\mathcal D}
\newcommand{\Ll}{\mathcal L}
\newcommand{\rk}{\mathrm{rk}}
\newcommand{\wt}{\mathrm{wt}}

\newcommand{\dd}{\mathrm{d}}
\newcommand{\djr}{\mathrm{d}_j^r}

\newcommand{\supp}{\mathrm{supp}}

\begin{document}

\maketitle
\begin{abstract}
We define a notion of $r$-generalized column distances for the $j$-truncation of a convolutional code. Taking the limit as $j$ tends to infinity allows us to define $r$-generalized column distances of a convolutional code. We establish some properties of these invariants and compare them with other invariants of convolutional codes which appear in the literature.
\end{abstract}

\section{Introduction}

Convolutional codes play an important practical role, as they are used extensively to achieve reliable data transmission in digital video, mobile communications, satellite communications, and other applications. Their popularity comes mostly from the fact that maximum-likelihood soft-decision decoding can be performed very efficiently on convolutional codes. In spite of the fact that they play a central role in the applications, however, the mathematical theory of convolutional codes is not as well-developed as for other families of codes.

In \cite{GS} we proposed and studied a new definition of generalized weights for a convolutional code and defined optimal convolutional anticodes. In this paper, we continue our investigation of the mathematical structure of convolutional codes by introducing another family of invariants: the $r$-generalized column distances. A notion of $r$-generalized column distance for the $j$-truncation of a convolutional code was given by Cardell, Firer, and Napp in \cite{CFN17}, for the special case of noncatastrophic convolutional codes. Later in \cite{CFN19} the same authors modified their definition and introduced unrestricted generalized column distances of noncatastrophic convolutional codes, which they further studied in \cite{CFN20}. 

In this work, we extend the original definition from \cite{CFN17} to any convolutional code. We call this invariant the $(r,j)$-generalized column distance of the code. Then, by taking the limit of the $(r,j)$-generalized column distance of a code as $j$ tends to infinity, we define its $r$-generalized column distance. This produces new invariants of convolutional codes, whose basic properties we study in this paper. We also introduce $j$-equivalences and equivalences of convolutional codes and we show that $(r,j)$-generalized column distances are invariant under $j$-equivalence and $r$-generalized column distances are invariant under equivalence. Using these notions, we are also able to show that $r$-generalized column distances are invariant under isometry. Notice that $(r,j)$-generalized column distances are not invariant under isometry, not even in the special case of noncatastrophic convolutional codes, as was already observed in \cite{CFN17}. 
In addition, we investigate the relations between generalized column distances and other invariants of convolutional codes, including the unrestricted generalized column distances defined in \cite{CFN19} and the generalized weights defined in \cite{GS}.

The paper is organized as follows. In Section 2 we introduce $j$-equivalences and equivalences and we study their basic properties, including their relation with isometries and strong isometries. In Section 3 we define $(r,j)$-generalized column distances and $r$-generalized column distances. We then establish some crucial properties of these invariants, e.g., we show that they are invariant under $j$-equivalences, equivalences, and isometries (see Proposition \ref{prop:invarjeq}, Corollary \ref{corollary:equivinv}, and Theorem \ref{theorem:isominv} for the precise statements). In Section 4 we discuss the relation between generalized column distances and unrestricted generalized column distances and generalized weights.

Throughout the paper, $\Cc\subseteq\F_q[x]^n$ denotes a convolutional code, i.e., an $\F_q[x]$-submodule of $\F_q[x]^n$. We refer to \cite{GS} for a discussion of the choice of working in this level of generality. We denote by $k$ the rank of $\Cc$ and always assume that $1\leq k\leq n$. Further, we denote by $\delta$ its internal degree and by $\delta_1$ the memory of $\Cc$.

\section{Equivalences and $j$-equivalences}

In this section we define $j$-equivalences and equivalences of convolutional codes. We study their main properties and their relation with isometries of convolutional codes. We start by recalling the definition of isometry and of strong isometry.

\begin{definition}
An $\F_q[x]$-isomorphism of convolutional codes $\phi:\Cc_1\rightarrow\Cc_2$ is an \textbf{isometry} if $\wt(c)=\wt(\phi(c))$ for all $c\in\Cc_1$. If in addition $\deg(c)=\deg(\phi(c))$ for all $c\in\Cc_1$, then $\phi$ is a \textbf{strong isometry}.
\end{definition}

Isometries of convolutional codes have been classified by Gluesing-Luerssen. 

\begin{theorem}[{\cite[Theorem 4.1]{Glu}}]\label{theorem:isometric}
Let $\phi:\Cc_1\rightarrow\Cc_2\subseteq\F_q[x]^n$ be an isometry of convolutional codes. There exist a permutation matrix $P\in\mathrm{GL}_n(\F_q)$ and a diagonal matrix
$D = \mathrm{diag}(a_1x^{m_1},\dots,a_nx^{m_n})$ where
$a_1,\dots,a_n\in\F_q^{*}$ and $m_1,\dots,m_n\in\Z$ such that $\phi(c)=cPD$ for all $c\in\Cc$.
\end{theorem}

Given an element $c=(p_1(x),\dots,p_n(x))\in\F_q[x]^n$ where $p_{\ell}(x)=a_{\ell,0}+a_{\ell,1}x+a_{\ell,2}x^2+\dots$, we can express it as $c=\sum_{i=0}^{\infty} c[i]x^i$, where $c[i]=(a_{1,i},\dots,a_{n,i})\in\F_q^n$. We define $c_{[h,j]}$ as
$$c_{[h,j]}=\sum_{i=h}^{j} c[i]x^i.$$
Usually, $c_{[0,j]}$ is called the j-th \textbf{truncation} of $c$. For a vector space $V\subseteq\F_q[x]^n$, we denote by $V_{[0,j]}$ the vector space generated by the $j$-th truncation of all the elements in $V$. The space $V_{[0,0]}$ will be denoted also as $V[0]$, since it coincides with the evaluation of the vector space at $x=0$.

\begin{definition}
For each $j\in\N_{0}$, an $\F_q[x]$-isomorphism of convolutional codes $\phi:\Cc_1\rightarrow\Cc_2$ is called $j$-\textbf{equivalence} if $\phi:(\Cc_1)_{[0,j]}\rightarrow(\Cc_2)_{[0,j]}$ is a Hamming weight-equivalence, i.e., $\wt(c_{[0,j]})=\wt(\phi(c)_{[0,j]})$ for all $c\in\Cc_1$. We say that $\phi$ is an \textbf{equivalence} if it is a $j$-equivalence for all $j\in\N_0$.
\end{definition}

The next proposition follows easily from the definition. It collects some of the basic properties of $j$-equivalences.

\begin{proposition}\label{proposition:propertiesjequiv}
Let $j\in\N_{0}$ and let $\phi:\Cc_1\rightarrow\Cc_2$  be a $j$-equivalence. The following hold.
\begin{enumerate}
\item $\phi^{-1}:\Cc_2\rightarrow\Cc_1$ is a $j$-equivalence.
\item $\dim(\Cc_1)_{[h,i]}=\dim(\Cc_2)_{[h,i]}$ for any $0\leq h\leq i\leq j$. In particular, $\dim(\Cc_1[0])=\dim(\Cc_2[0])$.
\item Let $\Dd\subseteq\Cc$, then $\phi\restriction_{\Dd}$ is a $j$-equivalence. 
\item If $\psi:\Cc_2\rightarrow\Cc_3$ is a $j'$-equivalence, then $\psi\circ\phi$ is a $\min\{j,j'\}$-equivalence.
\item If $\psi:\Cc_3\rightarrow\Cc_1$ is a $j'$-equivalence, then $\phi\circ\psi$ is a $\min\{j,j'\}$-equivalence.
\item If $\psi:\Cc_3\rightarrow\Cc_4$ is a $j'$-equivalence, then $\psi\times\phi$ is a $\min\{j,j'\}$-equivalence.
\end{enumerate}
\end{proposition}

In the next lemma, we collect a few more facts on $j$-equivalences.

\begin{lemma}\label{lemma:decreasingequivalence}
Let $\Cc_1$ and $\Cc_2$ be convolutional codes, let $j\in\N_{0}$, and let $\phi:\Cc_1\rightarrow\Cc_2$ be a $j$-equivalence. Then:
\begin{enumerate}
\item $\phi$ is a $j'$-equivalence for $0\leq j'\leq j$.
\item $\phi$ induces a Hamming-weight equivalence $\phi:(\Cc_1)_{[h,i]}\rightarrow(\Cc_2)_{[h,i]}$
for $0\leq h\leq i\leq j$.\\ In particular, $\wt(c_{[i,i]})=\wt(\phi(c)_{[i,i]})$ for $0\leq i\leq j$.
\end{enumerate}
\end{lemma}

\begin{proof}
To prove the first part of the thesis, it suffices to show that if $\phi$ is a $j$-equivalence then it is also a $(j-1)$-equivalence. Let $c\in\Cc_1$. Since $\wt(c_{[0,j-1]})=\wt((xc)_{[0,j]})$, we have that
$$\wt(c_{[0,j-1]})=\wt((xc)_{[0,j]})=\wt(\phi(xc)_{[0,j]})=\wt((x\phi(c))_{[0,j]})=\wt(\phi(c)_{[0,j-1]}).$$
It follows that $\phi$ is a $(j-1)$-equivalence.

To prove that the restriction $\phi:(\Cc_1)_{[h,i]}\rightarrow(\Cc_2)_{[h,i]}$ is a Hamming-weight equivalence for $0\leq h\leq i\leq j$, observe that for any $c\in\Cc_1$ $$\wt(c_{[h,i]})=\wt(c_{[0,i]})-\wt(c_{[0,h-1]})=\wt(\phi(c)_{[0,i]})-\wt(\phi(c)_{[0,h-1]})=\wt(\phi(c)_{[h,i]}).$$
\end{proof}

\begin{remark}
Fix $j\geq 0$. A $j$-equivalence may not be an isometry and vice versa, as the following examples show.
\begin{enumerate}
\item[(a)] Let $\Cc_1=\langle(1,x,1)\rangle_{\F_q[x]}$ and $\Cc_2=\langle(1,x,x)\rangle_{\F_q[x]}$. The $\F_q[x]$-linear map $\phi:\Cc_1\rightarrow\Cc_2$ given by $\phi((1,x,1))=(1,x,x)$ is an isometry, but not a $0$-equivalence. 
\item[(b)] Let $\Cc_1=\langle(1,x^2,x^3)\rangle_{\F_q[x]}$ and $\Cc_2=\langle(1,x^2,x^3+x^4)\rangle_{\F_q[x]}$. The $\F_q[x]$-linear map $\phi:\Cc_1\rightarrow\Cc_2$ given by $\phi((1,x^2,x^3))=(1,x^2,x^3+x^4)$ is a $3$-equivalence, but not an isometry. 
\end{enumerate}
\end{remark}

While a $j$-equivalence for a fixed value of $j$ may not be an isometry, every equivalence is a strong isometry.

\begin{proposition}\label{propositon:eqaresis}
An equivalence between convolutional codes is a strong isometry.
\end{proposition}

\begin{proof}
Let $\phi:\Cc_1\rightarrow\Cc_2$ be an equivalence and let $c\in\Cc_1$. Since $\phi$ is a $j$-equivalence for all $j\geq0$, we obtain that $$\wt(c)=\lim_{j\to\infty}\wt(c_{[0,j]})=\lim_{j\to\infty}\wt(\phi(c)_{[0,j]})=\wt(\phi(c)),$$
i.e., $\phi$ is weight-preserving. Moreover, $\wt(c_{[i,i]})=\wt(\phi(c)_{[i,i]})$ for all $i\geq 0$ by Lemma~\ref{lemma:decreasingequivalence}, hence  $\phi$ is degree-preserving.
\end{proof}

Notice that a strong isometry may not be an equivalence, as the next example shows.

\begin{example}
Let $\Cc_1=\langle (1,x^2)\rangle$, $\Cc_2=\langle (x,x^2)\rangle$. Then $\phi:\Cc_1\rightarrow\Cc_2$ defined as $\phi(p(x),x^2p(x))=(xp(x),x^2p(x))$ is a strong isometry which is not a $0$-equivalence, hence not an equivalence.
\end{example}

The next theorems provides us with a simple characterization of equivalences.

\begin{theorem}\label{theorem:equivalenceshape}
Let $\phi:\Cc_1\rightarrow\Cc_2$ be an equivalence of convolutional codes. There exist a permutation matrix $P\in\mathrm{GL}_n(\F_q)$ and a diagonal matrix $D=\mathrm{diag}(a_1,\dots,a_n)$ with $a_1\dots,a_n\in\F_q^*$ such that $\phi(c)=cPD$ for all $c\in\Cc_1$. In particular, every equivalence can be extended to an isometry of $\F_q[x]^n$.
\end{theorem}

\begin{proof}
The statement follows by combining Proposition~\ref{propositon:eqaresis}, Theorem~\ref{theorem:isometric}, and Lemma~\ref{lemma:decreasingequivalence}.
\end{proof}

The next proposition provides us with an effective criterion to check whether an isomorphism of convolutional codes is an equivalence.

\begin{proposition}\label{proposition:increasingequivalence}
Let $\phi:\Cc_1\rightarrow \Cc_2$ be an $\F_q[x]$-isomorphism of convolutional codes. Let $c_1,\dots,c_k$ be a basis of $\Cc_1$ and let $t=\max\{\deg(c_1),\dots,\deg(c_k),\deg(\phi(c_1)),\dots,\deg(\phi(c_k))\}$. If $\phi$ is a $j$-equivalence for some $j\geq t$, then it is an equivalence.
\end{proposition}

\begin{proof}
By Lemma~\ref{lemma:decreasingequivalence} it suffices to prove that, if $\phi$ is a $j$-equivalence for some $j\geq t$, then it is also a $(j+1)$-equivalence. This is equivalent to showing that $\wt(c_{[j+1,j+1]})=\phi(c_{[j+1,j+1]})$ for all $c\in\Cc$, since 
$$\wt(c_{[0,j+1]})=\wt(c_{[0,j]})+\wt(c_{[j+1,j+1]})=\wt(\phi(c)_{[0,j]})+\wt(c_{[j+1,j+1]})$$ and $$\wt(\phi(c)_{[0,j+1]})=\wt(\phi(c)_{[0,j]})+\wt(\phi(c)_{[j+1,j+1]}).$$
%In other words, we need to show that the restriction of $\phi$ is a Hamming-weight isometry between $\Cc_1[j+1]$ and $\Cc_2[j+1]$, where $\Cc_\ell[i]$ denotes the vector space $\langle c_{[i,i]}\mid c\in\Cc_\ell\rangle$, for $\ell=1,2$. The thesis now follows from observing that, since $j\geq t$, then $\Cc_\ell[j+1]=x\Cc_\ell[j]$ and $\phi(xc_{[0,j]})=x\phi(c_{[0,j]})$, since $\phi$ is a homomorphism of $\F_q[x]$-modules.
Let $c=p_1(x)c_1+\dots+p_k(x)c_k\in\Cc$ and let
$$\bar c=\frac{1}{x}((p_1(x)-p_1(0))c_1+\dots+(p_k(x)-p_k(0))c_k)\in\Cc.$$
Since $j\geq t=\max\{\deg(c_1),\dots,\deg(c_k),\deg(\phi(c_1)),\dots,\deg(\phi(c_k))\}$, then $c_{[j+1,j+1]}=x\overline{c}_{[j,j]}$ and $\phi(c)_{[j+1,j+1]}=x\phi(\bar c)_{[j,j]}$. 
As a consequence, we obtain
$$\wt\left(c_{[j+1,j+1]}\right)=\wt\left(\bar c_{[j,j]}\right)=\wt\left(\phi(\bar c)_{[j,j]}\right)=\wt\left(\phi(c)_{[j+1,j+1]}\right),$$ which concludes the proof.
\end{proof}

\begin{corollary}\label{corollary:isometrytoequiv}
Let $\phi:\Cc_1\rightarrow \Cc_2$ be an isometry of convolutional codes. Let $\delta_1$ be the memory of~$\Cc$.  If $\phi$ is a $\delta_1$-equivalence, then 
it is an equivalence.
\end{corollary}

\begin{proof}
Let $c_1,\dots,c_k$ be a row reduced basis of $\Cc_1$ such  that $\delta_1=\deg(c_1)\geq\deg(c_2)\geq\dots\geq\deg(c_k)$. Since $\phi$ is an isometry and a $\delta_1$-equivalence, we obtain that $\deg(\phi(c_i))\leq\delta_1$ for $1\leq i\leq k$. We conclude by Proposition~\ref{proposition:increasingequivalence}.
\end{proof}

\begin{remark}
When $\delta_1=0$, Corollary~\ref{corollary:isometrytoequiv} yields the MacWilliams Extension Theorem for linear block codes. Notice however that this is not a new proof of the MacWilliams Extension Theorem, as the proof of Theorem~\ref{theorem:isometric} relies on it.
\end{remark}

The next theorem provides us with a partial characterization of $j$-equivalences.

\begin{theorem}\label{theorem:characjequiv}
Let $\phi:\Cc_1\rightarrow\Cc_2$ be a $j$-equivalence of convolutional codes. There exist a permutation matrix $P\in\mathrm{GL}_n(\F_q)$ and a diagonal matrix $D=\mathrm{diag}(a_1,\dots,a_n)$ with $a_1\dots,a_n\in\F_q^*$ such that $\phi(c)_{[0,j]}=c_{[0,j]}PD$ for all $c\in\Cc_1$. 
\end{theorem}

\begin{proof}
Let $\Cc_1=\langle c_1,\dots, c_k\rangle_{\F_q[x]}$ and define $\bar\Cc_1,\bar\Cc_2$ as
$$\bar\Cc_1=\langle (c_1)_{[0,j]},\dots,(c_k)_{[0,j]}\rangle_{\F_q[x]}\text{ and }\bar\Cc_2=\langle \phi(c_1)_{[0,j]},\dots,\phi(c_k)_{[0,j]}\rangle_{\F_q[x]}.$$
Notice that for every $c\in\Cc_1$ there exists $\bar c\in\bar\Cc_1$ such that $\bar c_{[0,j]}=c_{[0,j]}$. Indeed, we have that
\begin{equation}\label{equation:cbar}
\begin{split}
     c_{[0,j]}&=\left(\sum_{i=1}^kp_i(x)c_i\right)_{[0,j]}=\left(\sum_{i=1}^kp_i(x)((c_i)_{[0,j]}+(c_i)_{[j+1,\deg(c_i)]})\right)_{[0,j]}\\
     &=\left(\sum_{i=1}^kp_i(x)((c_i)_{[0,j]})\right)_{[0,j]}=\bar c_{[0,j]}.
\end{split}
\end{equation}
Moreover, the $j$-equivalence $\phi$ induces an $\F_q[x]$-linear isomorphism $\bar\phi$ between $\bar\Cc_1$ and $\bar\Cc_2$ defined by $\bar\phi((c_i)_{[0,j]})=\phi(c_i)_{[0,j]}$ for $1\leq i\leq k$. Since $\phi$ is a $j$-equivalence, then
\begin{equation*}
\begin{split}
    \wt\left(\bar c_{[0,j]}\right)&=\wt\left(\left(\sum_{i=1}^kp_i(x)(c_i)_{[0,j]}\right)_{[0,j]}\right)=\wt\left(\left(\sum_{i=1}^kp_i(x)c_i\right)_{[0,j]}\right)\\
    &=\wt\left(\left(\sum_{i=1}^kp_i(x)\phi(c_i)\right)_{[0,j]}\right)=\wt\left(\left(\sum_{i=1}^kp_i(x)\phi(c_i)_{[0,j]}\right)_{[0,j]}\right)=\wt\left(\phi(\bar c)_{[0,j]}\right), 
\end{split}
\end{equation*}
and therefore also $\bar\phi$ is a $j$-equivalence. Moreover, since
$$j\geq \max\{\deg((c_1)_{[0,j]}),\dots,\deg((c_k)_{[0,j]}),\deg(\phi(c_1)_{[0,j]}),\dots,\deg(\phi(c_k)_{[0,j]})\},$$ then $\bar\phi$ is an equivalence by Proposition~\ref{proposition:increasingequivalence}.
By Theorem~\ref{theorem:equivalenceshape} there exist a permutation matrix $P\in\mathrm{GL}_n(\F_q)$ and a diagonal matrix $D=\mathrm{diag}(a_1,\dots,a_n)$ with $a_1\dots,a_n\in\F_q^*$ such that $\phi(\bar c)=\bar cPD$ for all $\bar c\in\bar\Cc_1$. For $c\in\Cc$, following the notation of \eqref{equation:cbar}, we have that
\begin{equation*}
    \begin{split}
        \phi(c)_{[0,j]}&=\left(\sum_{i=1}^kp_i(x)\phi(c_i)\right)_{[0,j]}=\left(\sum_{i=1}^kp_i(x)\phi(c_i)_{[0,j]}\right)_{[0,j]}\\
        &=\phi(\bar c)_{[0,j]}=(\bar cPD)_{[0,j]}=\bar c_{[0,j]}PD=c_{[0,j]}PD,
    \end{split}
\end{equation*}
which concludes the proof.
\end{proof}

\section{Generalized column distances}

A generator matrix $G$ of an $(n,k,\delta)$ convolutional code $\Cc$ can be expressed as $G=\sum G_ix^i$ with $G_i\in\F^{k\times n}$. The $j$-\textbf{truncated sliding generator} matrix $G_j^c$ is defined as
$$G_j^c=\begin{pmatrix}
G_0&G_1&\dots& G_j\\
0&G_0&\dots& G_{j-1}\\
\vdots&\ddots&&\vdots\\
0&\dots&0&G_0
\end{pmatrix}\in\F_q^{k(j+1)\times n(j+1)}.$$
For $j\geq\delta_1$, we define the matrix ${G_j^c}'$ as
\begin{equation*}
    {G_j^c}'=\begin{pmatrix}
G_0&G_1&\dots& G_{\delta_1}&0&\dots&0\\
0&G_0&\dots& G_{\delta_1-1}&G_{\delta_1}&\ddots&0\\
\vdots&\ddots&&\vdots&\ddots&\ddots&\vdots\\
0&\dots&0&G_0&G_1&\dots&G_{\delta_1}
\end{pmatrix}\in\F_q^{k(j+1)\times n(j+1+\delta_1)}.
\end{equation*}
The $j$-\textbf{truncated} code of $\Cc$ is $$\Cc(j)=\{(v^0,v^1,\dots,v^j)G_j^c:v^i\in\F_q^k,v^0\neq 0\}\subseteq\F_q^{n(j+1)}.$$
Equivalently, $$\Cc(j)=\{(p(x)G)_{[0,j]} : p(x)\in\F_q[x]^k, p(0)\neq 0\}\subseteq\Cc.$$
Notice that it may happen that $0\not\in\Cc(j)$, so $\Cc(j)$ is not a vector space in general. However, one always has $\Cc(0)\cup\{0\}=\Cc[0]\subseteq\F_q^n$. 

For every $j\geq0$ the $j$-th \textbf{column distance} of a convolutional code $\Cc$ is defined as
\begin{equation*}
    \dd_j^c(\Cc)=\min\wt\, \Cc(j)=\min\left\{\wt\left((v^0,v^1,\dots,v^j)G_j^c\right):v^i\in\F_q^k,v^0\neq 0\right\}.
\end{equation*}
It follows from the definition that $d_j^c(\Cc)\leq d_{j+1}^c(\Cc)$ for $j\geq 0$.

When the code is noncatastrophic, we have the following equivalent formulation
\begin{equation*}
\dd_j^c(\Cc)=\min\left\{\wt\left(c_{[0,j]}(x)\right):c(x)\in\Cc\text{ and }c(0)\neq0\right\}.
\end{equation*}
%Column distances were introduced in \cite{}
It follows that $d_j(\Cc)>0$ %$0<d_j^c(\Cc)\leq d_{j+1}^c(\Cc)$ 
for a noncatastrophic code $\Cc$ and $j\geq 0$.
In the case of catastrophic convolutional codes, instead, it may be that $\dd_j^c(\Cc)=0$ for some $j\geq 0$. However, the following proposition shows that for every $\Cc$ there exists a $\hat j$ such that $\dd_{\hat j}^c(\Cc)>0$.

\begin{proposition}\label{proposition:coldist>zero}
Let $\Cc$ be an $(n,k,\delta)$ convolutional code. There exists $\hat j\geq 0$ such that $\dd_{\hat j}^c(\Cc)>0$.
\end{proposition}

\begin{proof}
Let $c_1,\dots,c_k$ be a row reduced basis for $\Cc$. Since $d_j^c(\Cc)$ is a weakly increasing function of $j$, the thesis is equivalent to $d_j^c(\Cc)>0$ for $j$ sufficiently large. If $\dd_j^c(\Cc)=0$ for every $j\geq\delta_1$, then for every $j\geq\delta_1$ there exists $v_j\in\F_q^{k(j+1)}$ such that $v_j^0\neq 0$ and $v_jG_j^c=0$. 
By construction, the first $n(j+1)$ entries of $v_j{G_j^c}'$ are equal to $0$. By the pigeonhole principle there exist $\delta_1\leq j_1<j_2$ such that the last $n\delta_1$ entries of $v_{j_1}{G_{j_1}^c}'$ coincide with the last $n\delta_1$ entries of $v_{j_2}{G_{j_2}^c}'$. This implies that there exists two elements $s_1,s_2\in\Cc$ (corresponding to $v_{j_1}{G_{j_1}^c}'$ and $v_{j_2}{G_{j_2}^c}'$), polynomials $p_1,\dots,p_k,q_1,\dots,q_k$, and two indices $1\leq i_1,i_2\leq k$ such that $s_1=\sum p_ic_i$, $s_2=\sum q_ic_i$, $p_{i_1}(0)\neq0$, $q_{i_2}(0)\neq0$, and $x^{ j_2-j_1}s_1=s_2$. Since
$$x^{ j_2-j_1}\sum_{i=1}^k p_ic_i=\sum_{i=1}^k q_ic_i,\;\text{ then }\;\sum_{i=1}^k(x^{ j_2-j_1}p_i-q_i)c_i=0.$$
As $q_{i_2}(0)\neq 0$, we have that $x^{ j_2-j_1}p_{i_2}-q_{i_2}\neq0$. This contradicts the assumption that $c_1,\ldots,c_k$ is a basis of $\Cc$. We conclude that there exists a $\hat j$ such that $\dd_{\hat j}^c(\Cc)>0$.
\end{proof}

In~\cite{CFN17}, the authors introduce the concept of generalized column distances for noncatastrophic convolutional codes. In this paper, we extend their definition to arbitrary codes and establish some properties of these invariants.

\begin{definition}
Let $\Cc$ be an $(n,k,\delta)$ convolutional code and $G$ a generator matrix for $\Cc$. For every $j\geq0$ and $1\leq r\leq k$ we define the $(r,j)$-\textbf{generalized column distance} as
\begin{equation*}
\djr(\Cc)=\min\left\{\lvert\supp\{v_1G_j^c,\dots,v_rG_j^c\}\rvert:v_i\in\F_q^{k(j+1)}\text{ and }\dim\left(\langle v_1^0,\dots,v_r^0\rangle_{\F_q}\right)=r\right\}.
\end{equation*}
We say that $v_1,\ldots,v_r\in\F_q^{k(j+1)}$ {\bf realize} the $(r,j)$-generalized column distance with respect to $G$ if $\dim\left(\langle v_1^0,\dots,v_r^0\rangle_{\F_q}\right)=r$ and $\djr(\Cc)=\lvert\supp\{v_1G_j^c,\dots,v_rG_j^c\}\rvert$.

Finally, we define the $r$-\textbf{generalized column distance} as
\begin{equation*}
\dd^r(\Cc)=\lim_{j\to\infty}\djr(\Cc).
\end{equation*}
\end{definition}
Well-definedness of the $r$-generalized column distance follows from items 4 and 7 in Proposition~\ref{proposition:basicproperties}.
We stress that the generalized column distances do not depend on the choice of a generator matrix. Indeed, let $c_1\dots,c_k$ be the rows of a generator matrix $G$ of $\Cc$. Then
\begin{equation*}
\begin{split}
    \djr(\Cc)=\min\{\lvert\supp(V_{[0,j]})\rvert:\; &V\subseteq\Cc\text{ is an }\F_q\text{-linear space with }\dim(V)=r\text{ and}\\
    &\text{for all } v=p_1c_1+\dots+p_kc_k\in V\setminus\{0\}\text{ there is }i\text{ such that }p_i(0)\neq0\}.
\end{split}
\end{equation*}
If $\tilde G$ is another generator matrix of $\Cc$ with rows $\tilde c_1,\dots,\tilde c_k$, then there exists a unimodular matrix $U$ such that $\tilde G=UG$, i.e., $\tilde c_i=u_{i,1}c_1+\dots+u_{i,k}c_k$ for $1\leq i\leq k$. For $v=p_1c_1+\dots p_kc_k=\tilde p_1\tilde c_1+\dots \tilde p_k\tilde c_k\in \Cc$, we have that
$v=\tilde p_1(u_{1,1}c_1+\dots+u_{1,k}c_k)+\dots+\tilde p_k(u_{k,1}c_1+\dots+u_{k,k}c_k)$, hence 
$$p_i=\sum_{s=1}^k (\tilde p_su_{s,i})$$
for $1\leq i\leq k$.
Therefore, if there exists $i$ such that $p_i(0)\neq0$, then there exists $s$ such that $\tilde p_s(0)\neq0$. 

While the $(r,j)$-generalized column distance of a code $\Cc$ does not depend on the choice of a generator matrix $G$ of $\Cc$, the vectors $v_1,\dots,v_r$ that realize $\djr(\Cc)$ depend on the choice of the matrix $G$, as the next example shows.

\begin{example}
Let $\Cc$ be the code generated by the matrix $G$ whose rows are $(1,0,1)$ and $(0,1,0)$. Then $d_1^2(\Cc)=3$ and it is realized by the vectors $(1,0,0,0)$ and $(0,1,0,0)$ with respect to $G$. The matrix $G'$ whose rows are $(1,x,1)$ and $(0,1,0)$ is also a generator matrix of $\Cc$ and $d_1^2(\Cc)$ is realized by $(1,0,0,-1)$ and $(0,1,0,0)$ with respect to $G'$.
\end{example}

\begin{remark}
The $(r,j)$-generalized column distances of a noncatastrophic convolutional code are defined in~\cite{CFN17} as
\begin{equation*}
    \bar\djr(\Cc)=\min\left\{\lvert\supp\{(c_1)_{[0,j]},\dots,(c_r)_{[0,j]}\}\rvert:c_i\in\Cc\text{ and }\dim\left(\langle c_1(0),\ldots,c_r(0)\rangle_{\F_q}\right)=r\right\}.
\end{equation*}
The set of which we take the minimum is always nonempty, as $\dim(\Cc[0])=\rk(G_0)=k$ for a noncatastrophic $\Cc$ of $\rk(\Cc)=k$. 
Moreover, for a noncatastrophic code $\Cc$
$$\dim\left(\langle v_1^0,\dots,v_r^0\rangle_{\F_q}\right)=r\iff\dim\left(\langle v_1G[0],\dots,v_rG[0]\rangle_{\F_q}\right)=r,$$ where $G[0]$ denotes the matrix $G_j^c$ evaluated at $0$. In other words, $v_iG[0]=v_i^0G_0$ for $1\leq i\leq r$.
It follows that $\djr(\Cc)=\bar\djr(\Cc)$. 
\end{remark}

If the code $\Cc$ is catastrophic, then $\djr(\Cc)$ is well-defined, while $\bar\djr(\Cc)$ may not be. 

\begin{example}
Let $\Cc=\langle (1,0),(0,x)\rangle_{\F_q[x]}$. By a straightforward computation, we have that $\dd_0^1(\Cc)=0$, $\dd_0^2(\Cc)=1$, $\dd_j^1(\Cc)=1$ and $\dd_j^2(\Cc)=2$ for all $j\geq1$. Therefore, $\dd^1(\Cc)=1$ and $\dd^2(\Cc)=2$. Notice that in this case $\bar\dd_0^2(\Cc)$ is not defined, as $\dim(\Cc[0])=1$.
\end{example}

In the next proposition we collect several basic properties of generalized column distances. In particular, we have that the $(1,j)$-generalized column distance is exactly the $j$-column distance of the code and that generalized column distances are non-decreasing in both $r$ and $j$. Items~1, 2, and~6 were proved in the noncatastrophic case in~\cite{CFN17}.

\begin{proposition}\label{proposition:basicproperties}
Let $\Cc$ be an $(n,k,\delta)$ convolutional code and let $\Dd$ be a subcode of $\Cc$. Then
\begin{enumerate}
\item $\dd_j^1(\Cc)=\dd_j^c(\Cc)$ for $j\geq0$.
\item $\dd_j^r(\Cc)\leq\dd_j^{r+1}(\Cc)$ for $j\geq 0$ and $1\leq r<k$ and the inequality is strict if $\Cc$ is noncatastrophic.
\item $\dd^r(\Cc)<\dd^{r+1}(\Cc)$ for $1\leq r<k$.
\item $\dd^r_{j}(\Cc)\leq\dd_{j+1}^{r}(\Cc)$ for $j\geq 0$ and $1\leq r\leq k$.
\item $\djr(\Cc)\leq\djr(\Dd)$ for $j\geq0$ and $1\leq r\leq\rk(\Dd)$.
\item If $\Cc$ is noncatastrophic, then $\djr(\Cc)\leq(j+1)(n-k)+r$ for $j\geq 0$ and $1\leq r\leq k$.
\item $\djr(\Cc)\leq \dd^r(\Cc)\leq n(\delta_1+1)$  for $j\geq0$ and $1\leq r\leq k$, where $\delta_1$ is the memory of $\Cc$.
\end{enumerate}
\end{proposition}

\begin{proof}
Items 1, 4, 5 and the first part of item 2 follow directly from the definition. The noncatastrophic case of item 2 is shown in~\cite[Theorem 1]{CFN17}. For item 6, see~\cite[Proposition 1]{CFN17}. To prove item 7, it suffices to compute $\lvert\supp\{e_1G_j^c,\dots,e_kG_j^c\}\rvert$, where $e_i\in\F_q^{k(j+1)}$ is the $i$-th vector of the canonical basis.
In order to prove item 3, first notice that items 4 and 7 imply that $r$-generalized column distances are well-defined. Indeed, the limit always exists since the $(r,j)$-generalized column distances are non-decreasing in $j$, and it is finite by item 7. Moreover,
by item 7 there exists a $\bar j\in\N$ such that for all $m\in\N$ $\dd_{\bar j}^r(\Cc)=\dd_{\bar j+m}^r(\Cc)=\dd^r(\Cc)$. Fix $r>1$. Up to increasing $\bar j$, we may assume that $\dd_{\bar j}^{r-1}(\Cc)=\dd^{r-1}(\Cc)$. By Proposition~\ref{proposition:coldist>zero} there exist $\hat j$ such that $\dd_{\hat j}^c(\Cc)>0$. Let $j\geq\max\{\bar j,\hat j\}$ and suppose that $v_1,\dots,v_r$ realize $\djr(\Cc)$ with respect to $G$, i.e., $\djr(\Cc)=\lvert\supp\{v_1G_j^c,\dots,v_rG_j^c\}\rvert$ and $\dim\left(\langle v_1^0,\dots,v_r^0\rangle_{\F_q}\right)=r$. Since $j\geq \hat j$, we have that $v_iG_j^c\neq0$ for all $1\leq i\leq r$. Hence there exist $\alpha_2,\dots,\alpha_r\in\F_q$ such that $\djr(\Cc)>\lvert\supp\{(v_2-\alpha_2v_1)G_j^c,\dots,(v_r-\alpha_r v_1)G_j^c\}\rvert$ with $\dim\left(\langle (v_2-\alpha_2 v_1)^0,\dots,(v_r-\alpha_r v_1)^0\rangle_{\F_q}\right)=r-1$. We conclude that $\dd^r(\Cc)\geq\djr(\Cc)>\dd_{j}^{r-1}(\Cc)=\dd^{r-1}(\Cc)$.
\end{proof}

Notice in particular that, while $(r,j)$-generalized columns distances are only weakly increasing in $r$, $r$-generalized columns distances are strictly increasing in $r$.

\begin{example}
\begin{itemize}
\item [(a)] The $(r,j)$-generalized column distances may not be strictly increasing in $r$ for a fixed $j$, if the code is catastrophic. For instance, the code $\Cc=\langle (x,0),(0,x)\rangle$ has $\dd_0^1(\Cc)=\dd_0^2(\Cc)=0$. This is coherent with item 2 of Proposition~\ref{proposition:basicproperties}. On the other side, item 3 of Proposition~\ref{proposition:basicproperties} implies that, for $j$ large enough, the $(r,j)$-generalized column distances are strictly increasing with $r$, also in the catastrophic case. For example, the same code $\Cc$ has $\dd_j^1(\Cc)=1$ and $\dd_j^2(\Cc)=2$ for all $j\geq1$.
\item [(b)] The bound in item 6 of Proposition~\ref{proposition:basicproperties} may not hold for catastrophic codes. For instance, let $q$ be a prime and let $n<q$. Let $\Cc=\langle(1,1,\dots,1),(x,2x,\dots,nx)\rangle_{\F_q[x]}\subseteq\F_q[x]^n$. Then
$$\dd_1^2(\Cc)=2n-1> 2(n-2)+2=(j+1)(n-k)+r.$$
\end{itemize}
\end{example}

\begin{remark}
Specializing item 6 in Proposition~\ref{proposition:basicproperties} to $r=1$ and using item 1, we obtain the classical bound for column distances 
$$\dd_j^c(\Cc)\leq(n-k)(j+1) +1,$$
that was originally proved in~\cite{LRS}. Moreover, the bound is achieved for several triples of parameters. Indeed, if $\Cc$ is a noncatastrophic $(n,k,\delta)$ MDP code we immediately obtain by the definition and by item 3 that 
$\dd_j^r(\Cc)=(j+1)(n-k)+r$ for $0\leq j\leq \left\lfloor \frac{\delta}{k}\right\rfloor+\left\lfloor \frac{\delta}{n-k}\right\rfloor$.
\end{remark}

It has been already noticed in~\cite{Glu} that column distances are preserved neither under isometries nor under strong isometries. Item 1 of Proposition~\ref{proposition:basicproperties} therefore implies that $(r,j)$-generalized column distances cannot be invariant under isometries either. However, in the next proposition we show that they are invariant under $j$-isometries. Later in this section, we prove that $r$-generalized column distances are invariant under isometries. 

\begin{proposition}\label{prop:invarjeq}
Let $\phi:\Cc_1\rightarrow \Cc_2$ be a $j'$-equivalence. Then, $\djr(\Cc_1)=\djr(\Cc_2)$ for $1\leq r\leq k$ and $0\leq j\leq j'$.
\end{proposition}

\begin{proof}
Let $G$ be a generator matrix of $\Cc_1$ and suppose that $v_1,\dots,v_r\in\F_q^{k(j+1)}$ realize the $(r,j)$-generalized column distance of $\Cc$ with respect to $G$. Since $\phi$ is a $j'$-equivalence, by Theorem~\ref{theorem:characjequiv} there exist a permutation matrix $P\in\mathrm{GL}_n(\F_q)$ and a diagonal matrix $D=\mathrm{diag}(a_1,\dots,a_n)$ with $a_1\dots,a_n\in\F_q^*$ such that $\phi(G)_i=G_iPD$ for all $i\leq j'$. In particular,
\begin{equation*}
    \phi(G)_j^c=\begin{pmatrix}
G_0PD&G_1PD&\dots& G_jPD\\
0&G_0PD&\dots& G_{j-1}PD\\
\vdots&\ddots&&\vdots\\
0&\dots&0&G_0PD
\end{pmatrix}.
\end{equation*}
Hence
$$\djr(\Cc_2)\leq\lvert \supp\{v_1\phi(G)_j^c,\dots,v_r\phi(G)_j^c\}\rvert=\lvert \supp\{v_1G_j^c,\dots,v_rG_j^c\}\rvert=\djr(\Cc_1).$$
The reverse inequality follows by looking at $\phi^{-1}:\Cc_2\rightarrow\Cc_1$, which is a $j$-equivalence by item 1 in Propositon~\ref{proposition:propertiesjequiv}. 
\end{proof}

Since $\dd^r(\Cc)=\lim_{j\to\infty}\djr(\Cc)$, the previous proposition implies that $r$-generalized column distances are preserved by equivalences.
\begin{corollary}\label{corollary:equivinv}
    Let $\phi:\Cc_1\rightarrow \Cc_2$ be an equivalence. Then, $\dd^r(\Cc_1)=\dd^r(\Cc_2)$ for $1\leq r\leq k$. 
\end{corollary}

We now extend the result of Corollary~\ref{corollary:equivinv} to all isometries.

\begin{theorem}\label{theorem:isominv}
Let $\phi:\Cc_1\rightarrow \Cc_2$ be an isometry. Then, $\dd^r(\Cc_1)=\dd^r(\Cc_2)$ for $1\leq r\leq k$.
\end{theorem}

\begin{proof}
By Theorem~\ref{theorem:isometric} there exist a permutation matrix $P\in\mathrm{GL}_n(\F_q)$ and a diagonal matrix
$D = \mathrm{diag}(a_1x^{m_1},\dots,a_nx^{m_n})$ where
$a_1,\dots,a_n\in\F_q^{*}$ and $m_1,\dots,m_n\in\Z$ such that $\phi(c)=cPD$ for all $c\in\Cc$. Let $m=\max\{0,-m_i : 1\leq i\leq n\}$.
Let $G$ be a generator matrix of $\Cc_1$ and suppose that $p_1,\dots,p_r\in\F_q[x]^{k}$ realize the $(r,j+m)$-generalized column distance of $\Cc$ with respect to $G$, i.e., $\dim(\langle p_1(0),\ldots,p_r(0)\rangle_{\F_q}=r$ and $d_{j+m}^r(\Cc_1)=|\supp\{(p_1G)_{[0,j+m]},\ldots,(p_rG)_{[0,j+m]}\}|$. Then
\begin{equation}\label{eq:djr}
    \begin{split}
        \djr(\Cc_2)&\leq\left\lvert\supp\left\{(p_1\phi(G))_{[0,j]},\dots,(p_r\phi(G))_{[0,j]}\right\}\right\rvert=\left\lvert\supp\left\{(p_1GPD)_{[0,j]},\dots,(p_rGPD)_{[0,j]}\right\}\right\rvert\\
        &\leq\left\lvert\supp\left\{(p_1G)_{[0,j+m]},\dots,(p_rG)_{[0,j+m]}\right\}\right\rvert=\dd_{j+m}^r(\Cc_1).
    \end{split}
\end{equation}
In fact, the columns of $GPD$ are equal to those of $G$ up to permutation and multiplying by a constant and a power of $x$ with exponent smaller than or equal to $m$. Since none of these operations affects supports and a monomial of degree $t$ in the $i$-th column of $GP$ corresponds to a monomial of degree $t+m_i$ in the $i$-th column of $GPD$, then $(pG)_{[0,j+m]}$ contains all the monomials that appear in $(pGPD)_{[0,j]}$ (and possibly more). This proves the inequality in (\ref{eq:djr}).
As $j$ goes to infinity, we obtain
\begin{equation*}
    \dd^r(\Cc_2)=\lim_{j\to\infty}\djr(\Cc_2)\leq\lim_{j\to\infty}\dd_{j+m}^r(\Cc_1)=\dd^r(\Cc_1).
\end{equation*}
The reverse inequality follows by considering $\phi^{-1}$ instead of $\phi$. We conclude that the $r$-generalized column distances are invariant under isometries.
\end{proof}

It follows from Proposition~\ref{proposition:basicproperties} that the sequence $(\djr(\Cc))_{j\in\N}$ stabilizes after some $\bar j$. In the last part of this section, we give an upper bound on $\bar j$. 

\begin{theorem}\label{theorem:boundstabilze}
Let $\Cc$ be an $(n,k,\delta)$ convolutional code with memory $\delta_1$. If $\bar j=\min\{j:\djr(\Cc)=\dd^r(\Cc)\}$, then 
$$\bar j< [n(\delta_1+1)+1]q^{\delta_1 kr}.$$
\end{theorem}

\begin{proof}
If $\delta_1=0$ then the thesis holds, since $\dd^r(\Cc)=\dd_0^r(\Cc)$. Hence assume $\delta_1\geq 1$. Let $\bar j=\min\{j:\djr(\Cc)=\dd^r(\Cc)\}$ and let $v_1,\dots,v_r\in\F_q^{(\bar j+1)k}$ be vectors that realize $\dd_{\bar j}^r(\Cc)$ with respect to a row reduced generator matrix $G$ of $\Cc$. For each $\delta_1\leq j\leq \bar j$ let $\pi_j:\F_q^{(\bar j+1)k}\rightarrow \F_q^{( j+1)k}$ be the canonical projection on the first $(j+1)k$ entries. Define vector spaces $D_j$ as $$D_j=\langle \pi_j(v_1){G_j^c}',\dots,\pi_j(v_r){G_j^c}'\rangle_{\F_q}.$$ First of all notice that $(D_j)_{[0,j]}=(D_{\bar j})_{[0,j]}$, hence $\dd_{\bar j}^r(\Cc)\geq \left\lvert\supp(D_j)_{[0,j]}\right\rvert$. Consider now the sequence $\{(D_j)_{[j+1,j+\delta_1]}\}_{j\geq 0}$. Since there are at most $q^{\delta_1 kr}$ different $\F_q$-linear spaces in the sequence $\{(D_j)_{[j+1,j+\delta_1]}\}_{j\geq 0}$, for every $1\leq s\leq n(\delta_1+1)$ we can find two indices $j_1$ and $j_2$ such that
\begin{itemize}
    \item $\delta_1+(s-1)q^{\delta_1 kr}\leq j_1<j_2\leq \delta_1+sq^{\delta_1 kr}$
    \item $(D_{j_1})_{[{j_1}+1,j_1+\delta_1]}=(D_{j_2})_{[j_2+1,j_2+\delta_1]}$. 
\end{itemize}
We observe that if
$$\left\lvert\supp\left\{(D_{j_1})_{[0,{j_1}]}\right\}\right\rvert=\left\lvert\supp\left\{(D_{j_2})_{[0,{j_2}]}\right\}\right\rvert,$$
then $\dd^r(\Cc)=\dd_{j_1}^r(\Cc)$, contradicting the minimality of $\bar j$. Indeed, for $j=j_2+t(j_2-j_1)$ with $t\geq 0$, the vectors
$$t(v_i)=(v_i^0,\dots,v_i^{j_1},v_i^{j_1+1},\dots,v_i^{j_2},\underbrace{v_i^{j_1+1},\dots,v_i^{j_2},\dots,v_i^{j_1+1},\dots,v_i^{j_2}}_{\text{t times}})$$ realize $d_j^r(\Cc)$
and $\dd_{j_2+t(j_2-j_1)}^r(\Cc)=\dd_{j_1}^r(\Cc)$ for all $t$.
So the support must increase by at least $1$ every $q^{\delta_1 kr}$ steps. If $\bar j \geq [n(\delta_1+1)+1]q^{\delta_1 kr}$, then  $$\dd^r(\Cc)=\dd_{\bar j}^r(\Cc)>n(\delta_1+1),$$ 
which contradicts Proposition~\ref{proposition:basicproperties}. Hence we conclude that $\bar j < [n(\delta_1+1)+1]q^{\delta_1 kr}$. 
\end{proof}

Theorem~\ref{theorem:boundstabilze} implies that the $(r,j)$-generalized column distances can be computed in a finite amount of time by exhaustive search. As the parameters grow, however, such a computation quickly becomes practically infeasible.

\begin{remark}
Notice that we do not expect the bound in Theorem~\ref{theorem:boundstabilze} to be sharp. In fact, one gets a sharper bound by substituting the quantity $q^{\delta_1kr}$ by $\frac{q^{\delta_1 kr}-1}{(q^r-1)(q^r-q)\cdots(q^r-q^{r-1})}$, which is a tighter bound for the number of vector spaces of the form $(D_j)_{[j+1,j+\delta_1]}$. One therefore obtains $\bar j<[n(\delta_1+1)+1]\frac{q^{\delta_1 kr}-1}{(q^r-1)(q^r-q)\cdots(q^r-q^{r-1})}\sim n\delta_1q^{\delta_1kr-r^2}$. We do not know whether the bound can be further improved by means of different arguments.
\end{remark}

\section{Related definitions and comparison}

\begin{definition}\label{definition:unrestrected}
Let $\Cc$ be a noncatastrophic code. The $r$-th \textbf{unrestricted generalized column distance} of the $j$-truncated code of $\Cc$ is given by
\begin{equation}\label{equation:defunrestricted}
    \dd_r(\Cc(j))= \min\left\{\left\lvert\supp(D)\right\rvert: D\subseteq\Cc(j),\,%\left\lvert D\right\rvert=r\text{ and }
    \dim\left(\langle D\rangle_{\F_q}\right)=r\right\}
\end{equation}
for $1\leq r\leq k(j+1)$.
\end{definition}

The concept of unrestricted generalized column distances for noncatastrophic codes was introduced by Cardell, Firer, and Napp in~\cite{CFN19} and further studied in~\cite{CFN20} by the same authors. In particular they proved that the $r$-th unrestricted generalized column distance is strictly increasing as a function of $r$ and they showed how to compute it from the truncated parity-check matrix. 

If the code $\Cc$ is noncatastrophic, then $G_0$ has full rank, hence $G_j^c$ has full rank for every $j\geq 0$. Therefore
\begin{equation}\label{equation:defunrestrictedalternative}
    \dd_r(\Cc(j))= \min\left\{\left\lvert\supp\{v_1G_j^c,\dots,v_rG_j^c\}\right\rvert: v_i\in\F_q^{(j+1)k},\,v_1^0\neq0,\,\dim\left(\langle v_1,\dots,v_r\rangle_{\F_q}\right)=r\right\}.
\end{equation}
Notice that, unlike \eqref{equation:defunrestricted}, \eqref{equation:defunrestrictedalternative} allows to extend Definition~\ref{definition:unrestrected} to the catastrophic case in a similar way as we have done for the generalized column distances. Moreover \eqref{equation:defunrestrictedalternative} implies that, for every code~$\Cc$ and any $1\leq r\leq k$ and $0\leq j$, we have $\dd_r(\Cc(j))\leq\djr(\Cc)$.

If two codes have the same unrestricted generalized column distances, they may not have the same generalized column distances and vice versa, as the next examples show. The first example also shows how the limit as $j$ goes to infinity of the unrestricted generalized column distances may depend only on a submodule of $\Cc$ of rank strictly smaller than $k$.

\begin{example}\label{ex:4.2}
Let $\Cc_1=\langle (1,1,0,0,0),(0,0,1,1,1)\rangle_{\F_q[x]}$ and $\Cc_2=\langle (1,1,0,0,0),(0,x,1,1,1)\rangle_{\F_q[x]}$. One obtains by direct computation that
\begin{equation*}
    \dd_r(\Cc_1(j))=\dd_r(\Cc_2(j))=\begin{cases}2r&\text{if }1\leq r\leq j+1,\\
    2(j+1)+3(r-j-1)&\text{if }j+1<r\leq 2j+2.
    \end{cases}
\end{equation*}
On the other hand, one can check that $\dd_1^2(\Cc_1)=5$, while $\dd_1^2(\Cc_2)=6$. Moreover, $\dd^2(\Cc_1)=5$ and $\dd^2(\Cc_2)=6$, while for all $r\geq 1$ 
\begin{equation*}
    \lim_{j\to\infty}\dd_r(\Cc_1(j))=\lim_{j\to\infty}\dd_r(\Cc_1(j))=2r.
\end{equation*}
\end{example}

\begin{example}
Let $\Cc_1=\langle (1+x,1,0)\rangle_{\F_q[x]}$ and $\Cc_2=\langle (1,1,x)\rangle_{\F_q[x]}$. We have that $\dd_0^1(\Cc_1)=\dd_0^1(\Cc_2)=2$ and $\dd_j^1(\Cc_1)=\dd_j^1(\Cc_2)=3$ for all $j\geq1$. On the other side, $\dd_2(\Cc_1(1))=4$ and $\dd_2(\Cc_2(1))=5$.
\end{example}

For a fixed code and a fixed $r$, the $r$-th unrestricted generalized column distance of the $j$-truncated code may be increasing, decreasing or constant in $j$. This is different from the behavior of $(r,j)$-column distances which are non-decreasing in $j$ for a fixed $r$, as shown in item 4 in Proposition~\ref{proposition:basicproperties}.

\begin{example}
Let $\Cc_1=\langle (1,1,0,0,0),(0,0,1,1,1)\rangle_{\F_q[x]}$ be the code of Example~\ref{ex:4.2} and $\Cc_2=\langle (1,0,0),(0,x,1)\rangle_{\F_q[x]}$. It is easy to check that $$d_2(\Cc_1(0))=5>4=d_2(\Cc_1(1))=d_2(\Cc_1(2))\;\mbox{ and }\; d_4(\Cc_1(0))=4<5=d_4(\Cc_2(1)).$$
\end{example}

We now discuss the relation among $r$-generalized column distances and generalized weights. For an $\F_q$-linear code $\Ll$ we denote by $\dd_r^H(\Ll)$ the $r$-generalized Hamming weight. We refer to~\cite{Wei} for the definition and the basic properties of generalized Hamming weights.

\begin{proposition}
Let $\Cc$ be an $(n,k,\delta)$ convolutional code. Then
\begin{itemize}
\item $\dd_0^r(\Cc)=\dd_r^H(\Cc[0])$ for $1\leq r\leq k$.
%\item If $\Cc$ is noncatastrophic, then $\dd_0^r(\Cc)=\dd_r^H(\Cc[0])$ for $1\leq r\leq k$.
\item If $\delta=0$, then $\dd^r(\Cc)=\djr(\Cc)=\dd_r^H(\Cc[0])$ for $1\leq r\leq k$ and $j\geq0$.
\end{itemize}
\end{proposition}

\begin{proof}
The statement follows from the fact that for every linear code $\Ll\subseteq\F_q^n$ one has $\dim(\Ll)=\rk\left(\langle \Ll\rangle_{\F_q[x]}\right)$.
\end{proof}

In~\cite{GS} we have defined the generalized weights of a convolutional code as follows.
\begin{definition}
Let $\Cc$ be an $(n,k,\delta)$ convolutional code. The $r$-th \textbf{generalized
weight} of $\Cc$ is
    \begin{equation*}
        \dd_r(\Cc) = \min\left\{\left\lvert\supp\left\{d_1,\dots,d_k\right\}\right\rvert: \rk\left(\langle d_1,\dots,d_r\rangle_{\F_q[x]}\right)\geq r\right\},
    \end{equation*}
for $1\leq r\leq k$.    
\end{definition}

The next theorem gives an equivalent definition for the $r$-generalized column distances of a noncatastrophic code. It will be useful to establish an inequality between generalized weights and $r$-generalized column distances.

\begin{theorem}\label{theorem:coldistreal}
Let $\Cc$ be a noncatastrophic code. Then
\begin{equation*}
        \dd^r(\Cc)= \min\left\{\left\lvert\supp\left\{v_1{G_j^c}',\dots,v_r{G_j^c}'\right\}\right\rvert:j\geq0,\,v_i\in\F_q^{(j+1)k}\text{ with }\dim(\langle v_1^0,\dots,v_r^0\rangle)=r\right\}.
\end{equation*}
\end{theorem}

\begin{proof}
By definition of $\dd^r(\Cc)$, the left-hand side is smaller than or equal to the right-hand side. So it suffices to prove the reverse inequality. By Proposition~\ref{proposition:basicproperties} there exists a $\bar j$ such that $\dd_{\bar j}^r(\Cc)=\dd^r(\Cc)$. Let $G$ be a generator matrix for $\Cc$ and let $v_1,\dots,v_r\in\F_q^{(A+\bar j+1)k}$ with $A=q^{\delta_1kr}$ be vectors that realize $\dd_{\bar j+A}^r(\Cc)$ with respect to $G$. Then there exist $\bar j\leq j_1<j_2\leq \bar j+A$ such that, for $1\leq i\leq r$, $(v_i^0,\dots v_i^{j_1}){G_{j_1}^{c'}}$ corresponds to a codeword of the form $c_i+x^md_i$ and $(v_i^0,\dots v_i^{j_2}){G_{j_2}^{c'}}$ corresponds to a codeword of the form $c_i+x^nd_i$, where $n>m> \bar j$, $\deg(c_i)\leq \bar j$ and $\deg(d_i)\leq\delta_1-1$. Therefore,
$(x^n-x^m)d_i=c_i+x^nd_i-(c_i+x^md_i)\in\Cc$. Since the code is noncatastrophic, $d_i\in\Cc$ by~\cite[Proposition 3.1]{GS}, so there are $w_1,\dots,w_r\in\F_q^{(j_1+1)k}$ such that $w_1^0=\dots=w_r^0=0$ and $(v_i^0-w_i^0,\dots v_i^{j_1}-w_i^{j_1}){G_{j_1}^c}'$ corresponds to $c_i$. Hence
\begin{equation*}
\dd^r(\Cc)= \dd_{\bar j}^r(\Cc)=\left\lvert\supp\left\{(v_1^0-w_1^0,\dots v_1^{j_1}-w_1^{j_1}){G_{j_1}^{c'}},\dots,(v_r^0-w_r^0,\dots v_r^{j_1}-w_r^{j_1}){G_{j_1}^{c'}}\right\}\right\rvert.
\end{equation*}
This concludes the proof, since $\dim(\langle v_1^0,\dots,v_r^0\rangle)=r$ implies that $\dim(\langle(v_1-w_1)^0,\dots,(v_r-w_r)^0\rangle)=r$.
\end{proof}

If the code is catastrophic, the inequality
\begin{equation*}
        \dd^r(\Cc)\leq\min\{\lvert\supp\{v_1{G_j^{c'}},\dots,v_r{G_j^{c'}}\}\rvert:j\geq0,\,v_i\in\F_q^{(j+1)k}\text{ with }\dim(v_1^0,\dots,v_r^0)=r\},
\end{equation*} still holds. However we do not always have equality, as the next example shows.

\begin{example}\label{example:catgencoldis}
Consider the convolutional code $\Cc=\langle(1,x),(x,1)\rangle_{\F_q[x]}$. We have that $(1-x^2,0)=(1,x)-x(x,1)\in\Cc$. Similarly, $(1-x^{2n},0),(0,1-x^{2n})\in\Cc$. Therefore $\djr(\Cc)=\dd^r(\Cc)=r$ for all $j\geq 0$ and $r=1,2$. On the other side,
$$\min\left\{\left\lvert\supp\left\{v{G_j^c}'\right\}\right\rvert:j\geq0,\,v\in\F_q^{(j+1)k}\text{ with }v_1^0\neq0\right\}=2.$$
In particular, $\dd_1(\Cc)=2$.
\end{example}

For a noncatastrophic code, the $r$-th generalized weight is smaller than or equal to the $r$-generalized column distance. In addition, equality holds for $r=1$.

\begin{corollary}\label{corollary:ineqgenweig}
Let $\Cc$ be a noncatastrophic code. Then $$\dd_r(\Cc)\leq\dd^r(\Cc).$$ 
Moreover, $\dd_1(\Cc)=\dd^1(\Cc)=\dd_{\mathrm{free}}(\Cc)$.
\end{corollary}

\begin{proof}
Let $G$ be a generator matrix for $\Cc$ and let $c_1,\dots,c_k$ be the rows of $G$.
By Theorem~\ref{theorem:coldistreal} there exist $j\geq 0$ and $v_1,\dots,v_r\in\F_q^{(j+1)k}$ such that $\dd^r(\Cc)=\left\lvert\supp\left\{v_1{G_j^c}',\dots,v_r{G_j^c}'\right\}\right\rvert$ and $\dim(\langle v_1^0,\dots,v_r^0\rangle)=r$.
For each $1\leq i\leq r$, let $d_i\in\Cc$ be the element that corresponds to $v_i{G_j^c}'$. In order to conclude, it suffices to show that $\dim(\langle v_1^0,\dots,v_r^0\rangle)=r$ implies $\rk\left(\langle d_1,\dots,d_r\rangle_{\F_q[x]}\right)=r$. Since the code is noncatastrophic, it suffices to show that $\sum q_id_i\neq0$ for every set of polynomials $q_1,\dots,q_r$ such that at least one of them is not divisible by $x$. Suppose that $x\nmid q_s$. Since $\dim(\langle v_1^0,\dots,v_r^0\rangle)=r$, then
$$q_s(0)v_s^0G_0+\sum_{i\neq s}q_i(0)v_i^0G_0\neq0,$$
which implies that
\begin{equation*}
    q_s(0)d_s(0)+\sum_{i\neq s}q_i(0)d_i(0)=\sum_{i=1}^rq_i(0)d_i(0)\neq0.
\end{equation*}
In particular, $\sum q_id_i\neq0$. Finally,
\begin{equation*}
    \dd_{\mathrm{free}}(\Cc)\leq \dd^1(\Cc)\leq\dd_1(\Cc)=\dd_{\mathrm{free}}(\Cc),
\end{equation*}
where the first inequality follows from Theorem~\ref{theorem:coldistreal}.
\end{proof}

Example~\ref{example:catgencoldis} shows that for a catastrophic code one can have $\dd^r(\Cc)<\dd_r(\Cc)$. The next example shows that there are catastrophic codes for which $\dd_r(\Cc)<\dd^r(\Cc)$.

\begin{example}
Let $\Cc=\langle(1,x,0),(0,1,1)\rangle_{\F_q[x]}$. One can check that $\Cc$ is noncatastrophic and $\dd_2(\Cc)=3<4=\dd^2(\Cc)$.
\end{example}

We conclude this section by exhibiting families of codes whose $r$-th generalized weight coincides with their $r$-generalized column distance for all $1\leq r\leq k$. The proof of the next proposition is similar to the one of~\cite[Proposition 5.6]{GS}.

\begin{proposition}\label{proposition:coldisboundMDS}
Let $\Cc$ be a noncatastrophic MDS $(n,k,\delta)$ convolutional code. If $k\nmid \delta$, let $0<a<k$ such that $\delta=k\left\lceil\frac{\delta}{k}\right\rceil-a$. Then, for $1\leq r\leq a$
\begin{equation*}
\dd^r(\Cc)=\dd_r(\Cc)=(n-k)\left(\left\lfloor\frac{\delta}{k}\right\rfloor+1\right)+\delta+r,
\end{equation*}
and for $a<r\leq k$
\begin{equation*}
\dd^r(\Cc)\leq \dd_r^H(\Cc[0])+(n-k)\left(\left\lfloor\frac{\delta}{k}\right\rfloor+1\right)+\delta.
\end{equation*}
If $k\mid\delta$, then $$\dd^r(\Cc)=\dd_r(\Cc)=(n-k)\left(\frac{\delta}{k}+1\right)+\delta+r.$$
\end{proposition}

\begin{proof}
Let $c_1,\dots,c_k$ be a basis for $\Cc$. Since $\Cc$ is MDS, if $k\nmid\delta$, we may assume that $\deg(c_1)=\dots=\deg(c_a)=\left\lfloor\frac{\delta}{k}\right\rfloor$ and $\deg(c_{a+1})=\dots=\deg(c_k)=\left\lceil\frac{\delta}{k}\right\rceil$, see~\cite[Proposition 5.6]{GS} for more details. For $r\leq a$, we conclude by observing that
$\lvert\supp\{c_1,\dots,c_a\}\rvert\leq \dd_{\mathrm{free}}(\Cc)+a$ and by Proposition~\ref{proposition:basicproperties}. If $r>a$, let $d_1,\dots,d_r$ be linearly independent elements of $\langle c_1,\dots,c_k\rangle_{\F_q}$ such that $\lvert \supp\{d_1(0),\dots,d_r(0)\}\rvert=\dd^H_r(\Cc[0])$ and $\dim(\langle d_1(0),\dots,d_r(0)\rangle_{\F_q})=r$. Up to using $xc_1,\dots,xc_a$ to reduce the support, we may assume that $\lvert \supp\{(d_1)_{[\delta_1,\delta_1]},\dots,(d_r)_{[\delta_1,\delta_1]}\}\rvert\leq n-a$. Then
\begin{equation*}
    \lvert \supp\{d_1,\dots,d_r\}\rvert\leq \dd_r^H(\Cc[0])+n\left(\left\lfloor\frac{\delta}{k}\right\rfloor\right)+n-a=\dd_r^H(\Cc[0])+(n-k)\left(\left\lfloor\frac{\delta}{k}\right\rfloor+1\right)+\delta.
\end{equation*}
If $k\mid \delta$, then 
$\deg(c_1)=\dots=\deg(c_k)=\frac{\delta}{k}$ and 
\begin{equation*}
\left\lvert\supp\{c_1,\dots,c_k\}\right\rvert\leq n\left(\frac{\delta}{k}+1\right)=(n-k)\left(\frac{\delta}{k}+1\right)+\delta+k.
\end{equation*}
We conclude by items 1 and 3 of Proposition~\ref{proposition:basicproperties} and~\cite[Proposition 5.4]{GS}. 
\end{proof}

Notice that the generalized column distances of an MDS code are not determined by its parameters, as the next example shows.

\begin{example}
Let $\Cc=\langle(1,1,2),(2x,x+1,0)\rangle_{\F_q[x]}$. The code $\Cc$ is MDS and one can check that $\dd^2(\Cc)=5$. The reverse code $\mathrm{rev}(\Cc)$ of $\Cc$ is generated by $(1,1,2)$ and $(2,1+x,0)$. It is MDS, it has the same parameters as $\Cc$, and $\dd^2(\mathrm{rev}(\Cc))=4$. 
\end{example}

\begin{proposition}\label{propostion:mds+mdp}
Let $\Cc$ be a noncatastrophic convolutional code. If $\Cc$ is MDS and MDP, then $\dd_r(\Cc)=\dd^r(\Cc)$ for $1\leq r \leq k$.
\end{proposition}

\begin{proof}
If $k\mid\delta$, we conclude  by Proposition~\ref{proposition:coldisboundMDS}. So suppose that $k\nmid\delta$ and let $0<a<k$ be such that $\delta=k\left\lceil\frac{\delta}{k}\right\rceil-a$. Let $d_1,\dots,d_r$ be elements in $\Cc$ that realize $\dd_r(\Cc)$, that is, $d_1,\dots,d_r$ generate a subcode of rank $r$ and weight $\dd_r(\Cc)$. We refer to~\cite[Definition 3.3 and Definition 3.5]{GS} for the definition of weight of a code and of what it means for a set of elements of $\Cc$ to realize the $r$-th generalized weight. Since $\Cc$ is noncatastrophic, we may assume without loss of generality that $d_1(0)\neq 0$. If $\dim(\langle d_1(0),\dots,d_r(0)\rangle_{\F_q})=r$, then $d_r(\Cc)\geq d^r(\Cc)$ and we conclude by Corollary~\ref{corollary:ineqgenweig}. Otherwise, suppose that there exists $d\in\langle d_1,\dots,d_r\rangle_{\F_q}$ such that $d(0)=0$. Notice that $d_1^H(\Cc[0])=n-k+1$ since $\Cc$ is MDP. Then
\begin{equation*}
    \left\lvert\supp\{d_1,\dots, d_r\}\right\rvert\geq (n-k+1)+(n-k)\left(\left\lfloor\frac{\delta}{k}\right\rfloor+1\right)+\delta+r-1
\end{equation*}
Therefore, by Proposition~\ref{proposition:coldisboundMDS} and Corollary~\ref{corollary:ineqgenweig} 
\begin{equation*}
    \dd_r^H(\Cc[0])+(n-k)\left(\left\lfloor\frac{\delta}{k}\right\rfloor+1\right)+\delta=\dd_r(\Cc)\leq\dd^r(\Cc)\leq \dd_r^H(\Cc[0])+(n-k)\left(\left\lfloor\frac{\delta}{k}\right\rfloor+1\right)+\delta.
\end{equation*}
We conclude that in any case $\dd_r(\Cc)=\dd^r(\Cc)$.
\end{proof}

An $(n,k,\delta)$ strongly MDS code such that $(n-k)\mid\delta$ is both MDS and MDP. The next corollary then follows from Proposition~\ref{propostion:mds+mdp}.

\begin{corollary}
Let $\Cc$ be a noncatastrophic convolutional code. If $\Cc$ is strongly MDS and \mbox{$(n-k)\mid\delta$}, then $\dd_r(\Cc)=\dd^r(\Cc)$ for $1\leq r \leq k$.
\end{corollary}

\begin{comment}
The next example shows that being just MDS or MDP is not sufficient for Proposition~\ref{propostion:mds+mdp} to hold. 

\begin{example}
...
\end{example}
\end{comment}

\bibliographystyle{plain}
\bibliography{bib}

\begin{thebibliography}{1}

\bibitem{CFN17}
S.~D. Cardell, M.~Firer, and D.~Napp.
\newblock Generalized column distances for convolutional codes.
\newblock In {\em 2017 IEEE International Symposium on Information Theory
  (ISIT)}, pages 21--25, 2017.

\bibitem{CFN20}
S.~D. Cardell, M.~Firer, and D.~Napp.
\newblock Generalized column distances.
\newblock {\em IEEE Transactions on Information Theory}, 66(11):6863--6871,
  2020.

\bibitem{CFN19}
S.~D. Cardell, D.~Napp, and M.~Firer.
\newblock Unrestricted generalized column distances: A wider definition.
\newblock In {\em 2019 IEEE International Symposium on Information Theory
  (ISIT)}, pages 2783--2787, 2019.

\bibitem{Glu}
H.~Gluesing-Luerssen.
\newblock On isometries for convolutional codes.
\newblock {\em Advances in Mathematics of Communications}, 3:179--203, 2009.

\bibitem{LRS}
H.~Gluesing-Luerssen, J.~Rosenthal, and R.~Smarandache.
\newblock Strongly-{MDS} convolutional codes.
\newblock {\em IEEE Transactions on Information Theory}, 52(2):584--598, 2006.

\bibitem{GS}
E.~Gorla and F.~Salizzoni.
\newblock Generalized weights of convolutional codes.
\newblock {\em arXiv:2207.12170v1}, 2022.

\bibitem{Wei}
V.~K. Wei.
\newblock Generalized {H}amming weights for linear codes.
\newblock {\em IEEE Transactions on Information Theory}, 37(5):1412--1418,
  1991.

\end{thebibliography}
\end{document}